\documentclass[11pt]{article}
\usepackage{amssymb}
\usepackage{amsmath,mathtools}
\usepackage{amsfonts}
\usepackage{amsthm}
\usepackage{dsfont}
\usepackage{amssymb}
\usepackage{mathpazo}
\usepackage{fullpage}
\usepackage{enumerate}
\usepackage[colorlinks=true,linkcolor=blue,citecolor=blue,urlcolor=blue]{hyperref}

\usepackage{titlesec}
\titleformat{\section}[hang]{\Large\bfseries\filright}{\thesection.}{.5em}{}
\titleformat{\subsection}[hang]{\large\bfseries\filright}{}{0em}{}
\titleformat{\subsubsection}[block]{\bfseries}{}{0em}{}

\theoremstyle{definition}
\newtheorem{question*}{Question}

\newcommand{\Tr}{\operatorname{Tr}}
\newcommand{\I}{\mathds{1}}

\newcommand{\complex}{\mathbb{C}}
\newcommand{\M}{M}
\newcommand{\U}[1]{\op{U}(#1)}
\newcommand{\op}{\operatorname}

\renewcommand{\d}{\mathrm{d}}

\newcommand{\ip}[2]{\langle #1 , #2\rangle}

\newcommand{\Bigip}[2]{\Bigl\langle #1, #2 \Bigr\rangle}

\newcommand{\abs}[1]{\lvert #1 \rvert}

\newcommand{\norm}[1]{\lVert #1 \rVert}

\theoremstyle{plain}
\newtheorem{theorem}{Theorem}

\newtheorem{corollary}[theorem]{Corollary}
\newtheorem{proposition}[theorem]{Proposition}

\theoremstyle{definition}

\newtheorem{example}[theorem]{Example}

\theoremstyle{remark}

\usepackage{authblk}

\begin{document}

\title{Twirling channels have minimal mixed-unitary rank}
\author[ ]{Mark Girard\textsuperscript{1,2} and Jeremy Levick\textsuperscript{3}}

\affil[1]{Institute for Quantum Computing, University of Waterloo\vspace{0.4mm}}

\affil[2]{School of Computer Science, University of Waterloo\vspace{0.4mm}}

\affil[3]{Department of Mathematics \& Statistics, University of
  Guelph\vspace{0.4mm}}

\date{\today}
\maketitle

\begin{abstract}
 For a positive integer $d$ and a unitary representation $\rho:G\rightarrow\U{d}$ of a compact group~$G$, the \emph{twirling channel} for this representation is the linear mapping $\Phi:\M_d\rightarrow\M_d$ defined as $\Phi(X)=\int_{G}\d \mu(g)\, \rho(g)X\rho(g^{-1})$ for every $X\in\M_d$, where $\mu$ is the Haar measure on $G$. Such channels are examples of \emph{mixed-unitary} channels, as they are in the convex hull of the set of unitary channels of a fixed size. By Carath\'eodory's theorem, these channels can always be expressed as a finite linear combination of unitary channels. We consider the \emph{mixed-unitary rank} twirling channels---which is the minimum number of distinct unitary conjugations required to express the channel as a convex combination of unitary channels---and show that the mixed-unitary rank of every twirling channel is always equal to its Choi rank, both of which are equal to the dimension of the von Neumann algebra generated by the representation. Moreover, we show how to explicitly construct minimal mixed-unitary decompositions for these types of channels and provide some examples.
 
\end{abstract}

\section{Introduction}

A \emph{quantum channel} is a linear mapping of matrices $\Phi:\M_n\rightarrow\M_m$ that is both completely positive and trace preserving, where $\M_n$ denotes the space of~$n\times n$ matrices with complex entries for a positive integer~$n$. It is well known that a mapping $\Phi:\M_n\rightarrow\M_m$ is completely positive if and only if it admits a \emph{Kraus representation} of the form  
\begin{equation}
  \label{eq:Kraus form}
  \Phi(X) = \sum_{k=1}^r A_k X A_k^{\ast}
\end{equation}
for every $X\in\M_n$, for some choice of positive integer $r$ and matrices $A_1,\dots,A_r\in\M_{m,n}$, where $\M_{m,n}$ denotes the space of $m\times n$ matrices \cite{Choi1975}. 
A map $\Phi$ described in this way preserves trace if and only if $\sum_{k=1}^r A_k^{\ast} A_k = \I_n$, where $\I_n$ is the $n\times n$ identity matrix.
The minimum value of $r$ for which a description of the form in~\eqref{eq:Kraus form} exists is called the
\emph{Choi rank} of~$\Phi$; this number being so-named because it is equal to
the rank of the \emph{Choi matrix} associated with $\Phi$, which is the matrix $J(\Phi)\in\M_m\otimes\M_n$ defined as
\begin{equation}
  J(\Phi) = \sum_{j,k=1}^n \Phi(E_{j,k}) \otimes E_{j,k},
\end{equation}
where $E_{j,k}\in\M_n$ denotes the matrix having a 1 in the $(j,k)$-entry 
and 0 in all other entries.

One of the simplest types of quantum channels are \emph{unitary channels}. These are channels of the form $\Phi:\M_n\rightarrow\M_n$ that are given by $\Phi(X)= UXU^*$ for every $X\in\M_n$, for some fixed choice of a unitary matrix $U\in\U{n}$, where $\U{n}$ denotes the group of~$n\times n$ unitary matrices.  A channel is a \emph{mixed-unitary channel} if it can be expressed as a convex combination of unitary channels. That is, a channel $\Phi:\M_n\rightarrow\M_n$ is mixed unitary if and only if
there exists a positive integer $N$, a probability vector $(p_1,\ldots,p_N)$,
and unitary matrices $U_1,\ldots,U_N\in\U{n}$ such that
\begin{equation}
  \label{eq:general mixed-unitary channel}
  \Phi(X) = \sum_{k = 1}^N p_k U_k X U_k^{\ast}
\end{equation}
for every $X\in\M_n$. The set of mixed-unitary channels of a fixed size is the convex hull of a compact set (the set of unitary channels) in a finite-dimensional space, and thus---by Carath\'eodory's theorem---every element in the closed convex hull of the set of unitary channels can be represented as a (finite) convex combination of unitary channels. For a mixed-unitary channel $\Phi:\M_n\rightarrow\M_n$, the \emph{mixed-unitary rank} of~$\Phi$ is the smallest number $N$ for which $\Phi$ has an expression of the form in~\eqref{eq:general mixed-unitary channel}. If $r$ is the Choi rank and $N$ is the mixed-unitary rank of a channel $\Phi$, it is obvious that $r\leq N$, but beyond that it is not straightforward to determine bounds on the mixed-unitary rank of a mixed-unitary channel. We shall say that a mixed-unitary channel has \emph{minimal} mixed-unitary rank if it is the case that $N=r$ (which is the smallest that $N$ can be). Recent work by the authors \cite{Girard2020} investigates properties of the mixed-unitary rank of general mixed-unitary channels, where upper bounds on the mixed-unitary rank of a channel are presented in terms of its Choi rank and the dimension of a corresponding operator system. Moreover, the work in \cite{Girard2020} presents the first known examples of mixed-unitary channels whose mixed-unitary ranks are not minimal. Previous bounds on the mixed-unitary rank of channels were presented in \cite{Buscemi2006}, while further properties of mixed-unitary channels have been studied in \cite{Audenaert2008} and \cite{Mendl2009}.

In this paper we are concerned with determining the mixed-unitary rank of certain types of mixed-unitary channels known as \emph{twirling channels}, which we describe as follows. Let $G$ be a compact group, let $\rho:G\rightarrow \U{n}$ be a unitary representation of $G$ for some positive integer~$n$, and let $\mu$ be the Haar measure on $G$. The twirling channel associated with this representation is the linear mapping $\Phi_\rho:\M_n\rightarrow\M_n$ defined as
\begin{equation}
 \Phi_\rho(X) = \int_G \d\mu(g)\, \rho(g)X\rho(g^{-1})
\end{equation}
for every $X\in\M_n$. It is evident that every twirling channel is mixed unitary and one may therefore consider the mixed-unitary rank of such a channel. One well-known example of a twirling channel used in the literature is the \emph{Werner twirling channel} $\Xi:\M_{n^2}\rightarrow\M_{n^2}$ defined as 
\begin{equation}\label{eq:wernertwirl}
 \Xi(X) = \int_{\U{n}} \d\mu(U)\,(U\otimes U)X(U\otimes U)^* 
\end{equation}
for every $X\in\M_{n^2}$, where $\mu$ is the Haar measure of the unitary group $\U{n}$. In the case when the group $G$ is finite, the corresponding twirling channel takes the form
\begin{equation}
 \Phi_\rho(X) = \frac{1}{\abs{G}}\sum_{g\in G}\rho(g)X\rho(g^{-1}).
\end{equation}
Twirling channels have a long history in the quantum information literature and have numerous applications. For example, channels of this form have been used in the contexts of quantum error correction \cite{Bennett1996}, quantum data hiding \cite{DiVincenzo2002}, as well as in the study of quantum entanglement \cite{Bennett1996a,Vollbrecht2001} and quantum coherence \cite{Chen2016a}.

The main result of this paper is summarized in the following statement: every twirling channels has minimal mixed-unitary rank. That is, the mixed-unitary rank of any twirling channel is equal to the Choi rank of that channel. This result is obtained through some simple observations regarding unitary representations and von Neumann algebras.

Before proceeding, we remark on the relation of our work to the concept of unitary designs \cite{Dankert2009,Roy2009}. Let $n$ and $t$ be positive integers and consider the linear mapping $\Xi:\M_{n^t}\rightarrow \M_{n^t}$ defined as
\begin{equation}
 \Xi(X) = \int_{\U{n}}\d\mu(U)\, \bigl(U^{\otimes t}\bigr)X\bigl(U^{\otimes t}\bigr)^*
\end{equation}
for every $X\in\M_{n^t}$. (Note that the channel in \eqref{eq:wernertwirl} is a channel of this form for the case $t=2$.) A collection of unitary matrices $U_1,\dots,U_N\in\U{n}$ for some positive integer $N$ comprises a \emph{unitary $t$-design} of dimension $n$ if it holds that 
\begin{equation}\label{eq:tdesign}
 \Xi(X) = \frac{1}{N}\sum_{k=1}^N \bigl(U_k^{\otimes t}\bigr)X\bigl(U_k^{\otimes t}\bigr)^*
\end{equation}
for every matrix $X\in\M_{n^t}$. The number $N$ is said to be the \emph{size} of this unitary $t$-design. The channel in \eqref{eq:wernertwirl} is clearly mixed unitary and an expression of the form in  \eqref{eq:tdesign} comprises a mixed-unitary decomposition of the channel $\Xi$ as each of the matrices $U_k^{\otimes t}\in\U{n^t}$ are unitary. The size of any unitary $t$-design must therefore be bounded below by the mixed-unitary rank of the channel in \eqref{eq:wernertwirl}. The mixed-unitary rank of $\Xi$ is, in general, distinctly different from the size of the smallest unitary $t$-design. Indeed, for the case when $t=2$, it is known that the size of any unitary 2-design must be at least $n^4-2n^2+2$ \cite{Gross2007}. Meanwhile, the mixed-unitary rank of the channel in \eqref{eq:wernertwirl} will be shown in this paper to be equal to $(n^4+n^2)/2$.

\section{Preliminaries}
In this section we summarize some of the notation as well as some known facts and results concerning quantum channels and representation theory of compact groups that will be used later in the paper. Further
information on quantum channels, and the role they play in the theory of quantum information,
can be found in texts on the subject, including \cite{Watrous2018}.

Given any matrix $A\in\M_{m,n}$, we denote by $A^{\ast}$ the adjoint (or conjugate
transpose) of $A$. The Hilbert-Schmidt inner product on $\M_{m,n}$ is defined as
$\ip{A}{B}  = \Tr(A^*B)$
for every choice of matrices $A, B \in \M_{m,n}$. A matrix $A\in\M_{m,n}$ is said to be an \emph{isometry} if $A^*A = \I_n$, where $\I_n$ is the $n\times n$ identity matrix.

Some simple examples of channels include the \emph{identity channel} and the \emph{completely depolarizing channel}, which we describe below. For a positive integer~$n$, the \emph{identity channel} of dimension~$n$ is the mapping $\I_{\M_n}:\M_n\rightarrow\M_n$ defined as $\I_{\M_n}(X)=X$ for every $X\in\M_n$ while the \emph{completely depolarizing channel} of dimension~$n$ is the linear mapping $\Omega_n:\M_n\rightarrow\M_n$ defined as
\begin{equation}
 \Omega_n(X) = \frac{1}{n}\Tr(X)\I_n
\end{equation}
 for every $X\in\M_n$. One choice of Kraus representation for the completely depolarizing channel is
\begin{equation}\label{eq:depolarizingkraus}
 \Omega_n(X) = \frac{1}{n}\sum_{j,k=1}^n E_{j,k} X E_{j,k}^*.
\end{equation}
The completely depolarizing has Choi representation $J(\Omega_n)=\tfrac{1}{n}\I_n\otimes\I_n$ and thus has Choi rank equal to~$n^2$. The mixed-unitary rank of the completely depolarizing channel is also equal to $n^2$, as one may express this channel as
\begin{equation}
 \Omega_n(X) = \frac{1}{n^2}\sum_{k=1}^{n^2}U_kXU_k^*
\end{equation}
for any orthogonal collection of unitary matrices $\{U_1,\dots,U_{n^2}\}\subseteq\U{n}$. (Such a collection exists, as one may choose, for example, the discrete Weyl operators.)

\subsection{Complementary channels}\label{sec:prelim}

For positive integers~$n$, $m$, and $r$, two channels $\Phi:\M_n\rightarrow\M_m$ and $\Psi:\M_n\rightarrow\M_r$ are said to be \emph{complementary} if there exists an isometry $A\in\M_{mr,n}$ such that
\begin{equation}
 \Phi(X) = (\I_{\M_m}\otimes \mathrm{Tr})(AXA^*) 
 \qquad\text{and}\qquad
 \Psi(X) = (\mathrm{Tr}\otimes\I_{\M_r})(AXA^*) 
\end{equation}
for every $X\in\M_n$, where one views the $mr\times mr$ matrix $AXA^*$ as an element of the tensor product space $\M_m\otimes M_r\simeq \M_{mr}$. Complementarity of channels can also be determined in terms of Kraus representations. For any choice of matrices $A_1,\dots,A_r\in\M_{m,n}$ satisfying $\sum_{k=1}^r A_k^*A_k=\I_n$, one may define a matrix $A\in\M_{mr,n}$ as
\begin{equation}
 A = \sum_{k=1}^r A_k\otimes e_k,
\end{equation}
where $\{e_1,\dots,e_n\}$ is the standard basis of $\complex^n$. This matrix is an isometry, as it may be verified that $A^*A = \sum_{k=1}^r A_k^*A_k$. For this choice of isometry, it holds that
\begin{equation}
 AXA^* = \sum_{j,k=1}^r A_jXA_k^* \otimes E_{j,k}
\end{equation}
for every $X\in\M_n$, and it follows that the channels $\Phi$ and $\Psi$ defined as
\begin{equation}
 \Phi(X)  =\sum_{k = 1}^r A_k X A_k^{\ast}
 \qquad\text{and}\qquad
 \Psi(X) = \sum_{j,k=1}^r \ip{A_j^{\ast} A_k}{X} \, E_{j,k}
\end{equation}
are complementary. If it is the case that $r=\op{rank}(J(\Phi))$ and $\Psi:M_{n}\rightarrow M_r$ is a channel that is complementary to $\Phi$, then any other given channel $\Xi:M_n\rightarrow M_N$ is also complementary to $\Phi$ if and only if there exists an isometry $V\in M_{N,r}$ such that $\Xi(X) = V\Psi(X)V^*$ for every $X \in M_n$. 

It will be useful to examine forms of channels that are complementary to the identity channel~$\I_{\M_n}$ and the completely depolarizing channel $\Omega_n$ for a fixed positive integer~$n$. It is evident that the trace mapping $\Tr:\M_n\rightarrow\complex$ is complementary to the identity channel. One choice of complementary channel to the completely depolarizing map is provided by the following proposition. 
\begin{proposition}
 Let $n$ be a positive integer and define a linear mapping $\Psi_n:\M_n\rightarrow\M_{n^2}$  as
 \begin{equation}
 \Psi_n(X) = \frac{1}{n}\I_n\otimes X
\end{equation}
for every $X\in\M_n$. Then $\Psi_n$ is a channel that is complementary to the completely depolarizing channel $\Omega_n$. 
\end{proposition}
\begin{proof}Define a matrix $A\in\M_{n^3,n}$ as
\begin{equation}
 A = \frac{1}{\sqrt{n}}\sum_{j,k=1}^n E_{j,k}\otimes e_j\otimes e_k.
\end{equation}
This is an isometry, as $A^*A = (\sum_{j,k=1}^n E_{k,j}E_{j,k})/n = \sum_{k=1}^nE_{k,k} = \I_n$. For this choice of isometry, one has that
\begin{equation}
 (\I_{M_n}\otimes\Tr)(AXA^*) = \frac{1}{n}\sum_{j,k=1}^n E_{j,k}XE_{k,j} = \Omega_n(X)
\end{equation}
and that
\begin{align*}
 (\Tr\otimes\I_{\M_{n^2}})(AXA^*) = \frac{1}{n}\sum_{i,j,k,l=1}^n \Tr(E_{j,k}XE_{l,i}) E_{j,i}\otimes E_{k,l} 
   &= \frac{1}{n}\sum_{j,k,l=1}^n \ip{E_{k,l}}{X}E_{j,j}\otimes E_{k,l}  \\&= \frac{1}{n}\I_n\otimes X = \Psi_n(X)
\end{align*}
for every $X\in\M_n$. It follows that the channels $\Omega_n$ and $\Psi_n$ are complimentary, as desired.
\end{proof}

We also note that complementarity of channels is well behaved under tensor products and direct sums. That is, if $\Psi_0:\M_{n_0}\rightarrow\M_{r_0}$ and $\Psi_1:\M_{n_1}\rightarrow\M_{r_1}$ are channels that are complementary to some channels $\Phi_0:\M_{n_0}\rightarrow\M_{m_0}$ and $\Phi_1:\M_{n_1}\rightarrow\M_{m_1}$ respectively, then the channels $\Psi_0\otimes\Psi_1$ and $\Psi_0\oplus\Psi_1$ are complementary to the channels $\Phi_0\otimes\Phi_1$ and $\Phi_0\oplus\Phi_1$ respectively. To see this, note that there must exist a choice of isometries $A_0\in\M_{m_0r_0,n_0}$ and $A_1\in\M_{m_1r_1,n_1}$ such that
\begin{equation}
 \Phi_0(X) = (\I_{\M_{m_0}}\otimes \mathrm{Tr})(A_0XA_0^*)\qquad\text{and}\qquad
 \Psi_0(X) = (\mathrm{Tr}\otimes\I_{\M_{r_0}})(A_0XA_0^*)
\end{equation}
for every $X\in\M_{n_0}$ and 
\begin{equation}
\Phi_1(Y) = (\I_{\M_{m_1}}\otimes \mathrm{Tr})(A_1YA_1^*)\qquad\text{and}\qquad
 \Psi_1(Y) = (\mathrm{Tr}\otimes\I_{\M_{r_1}})(A_1YA_1^*)
\end{equation}
for every $Y\in\M_{n_1}$. The matrices $A_0\otimes A_1\in\M_{m_0r_0m_1r_1,n_0n_1}$ and $A_0\oplus A_1\in\M_{(m_0r_0+m_1r_1),(n_0+n_1)}$ are also isometries. Moreover, for every $X\in\M_{n_0}$ and $Y\in\M_{n_1}$, one has that
\begin{align*}
 (\Phi_0\otimes\Phi_1)(X\otimes Y) &= \bigl(\I_{\M_{m_0}}\otimes \mathrm{Tr}\otimes \I_{\M_{m_1}}\otimes \mathrm{Tr}\bigr)\bigl((A_0\otimes A_1)(X\otimes Y)(A_0\otimes A_1)^*\bigr)\\
 (\Psi_0\otimes\Psi_1)(X\otimes Y) &= \bigl(\mathrm{Tr}\otimes \I_{\M_{r_0}}\otimes \mathrm{Tr}\otimes \I_{\M_{r_1}}\bigr)\bigl((A_0\otimes A_1)(X\otimes Y)(A_0\otimes A_1)^*\bigr)\\
 (\Phi_0\oplus\Phi_1)(X\oplus Y) &= \bigl((\I_{\M_{m_0}}\otimes \mathrm{Tr})\oplus( \I_{\M_{m_1}}\otimes \mathrm{Tr})\bigr)\bigl((A_0\oplus A_1)(X\oplus Y)(A_0\oplus A_1)^*\bigr)\\
 \text{and}\quad(\Psi_0\oplus\Psi_1)(X\oplus Y) &= \bigl((\mathrm{Tr}\otimes\I_{M_{r_0}})\oplus(\mathrm{Tr}\otimes\I_{M_{r_1}})\bigr)\bigl((A_0\oplus A_1)(X\oplus Y)(A_0\oplus A_1)^*\bigr).
\end{align*}

The preceding observations allow us to determine a complementary channel for channels of the following form. Let $p$ be a positive integer, let $m_1,\dots,m_p$ and~$n_1,\dots,n_p$ be positive integers, define the number $d=m_1n_1+\cdots+m_pn_p$, and define a channel $\Phi:\M_d\rightarrow\M_d$ as
\begin{equation}\label{eq:directsumphi}
 \Phi = (\I_{M_{m_1}}\otimes\Omega_{n_1})\oplus\cdots\oplus(\I_{M_{m_p}}\otimes\Omega_{n_p}).
\end{equation}
If one defines the number $N=n_1^2+\cdots+n_p^2$ and a linear mapping $\Psi:\M_d\rightarrow\M_N$ as 
\begin{equation}\label{eq:directsumpsi}
 \Psi\bigl((X_1\otimes Y_1)\oplus\cdots\oplus(X_p\otimes Y_p)\bigr) = \left(\frac{\Tr(X_1)}{n_1} \I_{n_1}\otimes Y_1\right)\oplus\cdots\oplus\left(\frac{\Tr(X_p)}{n_p} \I_{n_p}\otimes Y_p\right)
\end{equation}
for every $X_1\in\M_{m_1},\dots,X_p\in\M_{m_p}$ and $Y_1\in\M_{n_1},\dots,Y_p\in\M_{n_p}$, then $\Psi$ defined in \eqref{eq:directsumpsi} is a channel that is complementary to the channel $\Phi$ defined in \eqref{eq:directsumphi}. 

\subsection{Mixed-unitary channels and mixed-unitary rank}

We now review a characterization of the mixed-unitary rank of channels that will be useful for determining the mixed-unitary rank of twirling channels. The reader is referred to \cite{Girard2020} for a proof and further details. Note that a square matrix $X\in\M_n$ is said to be \emph{traceless} if $\Tr(X) =0$ and is said to have \emph{vanishing diagonal} if all of its diagonal entries are equal to $0$.

\begin{theorem}\label{thm:mixedunitarycomplement}
 Let~$n$ be a positive integer, let $\Phi:\M_n\rightarrow\M_n$ be a channel having Choi rank $r$, and let $N\geq r$ be an integer. The following statements are equivalent.
 \begin{enumerate}
  \item The channel $\Phi$ is mixed unitary with mixed-unitary rank at most $N$.
  \item There is a channel $\Psi:\M_n\rightarrow\M_N$ complementary to $\Phi$ such that $\Psi(X)$ has vanishing diagonal for every traceless matrix $X\in\M_n$.
 \end{enumerate}
\end{theorem}

We note that tensor products and direct sums of mixed-unitary channels are again mixed unitary. That is, provided that each of the channels $\Phi_1:\M_{n_1}\rightarrow\M_{n_1},\dots,\Phi_p:\M_{n_p}\rightarrow\M_{n_p}$ are mixed unitary for some positive integers~$n_1,\dots,n_p$, the channels
\[
 \Phi_1\otimes\cdots\otimes\Phi_p:\M_{n_1\cdots n_p}\rightarrow \M_{n_1\cdots n_p}
\]
and
\[
 \Phi_1\oplus\cdots\oplus\Phi_p:\M_{n_1+\cdots + n_p}\rightarrow \M_{n_1+\cdots+ n_p}
\]
are also mixed-unitary. (See \cite{Girard2020} for details.) As the identity channel $\I_{\M_n}$ and completely depolarizing channel $\Omega_n$ are mixed unitary for every positive integer~$n$, it follows that the channels of the form in~\eqref{eq:directsumphi} are mixed unitary. Although the Choi rank of channels of this form can be easily computed, determining the mixed-unitary rank of these channels is not so straightforward.

\subsection{Representation theory}\label{sec:representation}

In this section we outline some of the important facts from representation theory that will be used in this paper.  We refer the reader to various texts on the subject for further details, such as \cite{Knapp2001}, \cite{Goodman2009}, and \cite{Sagan2001}.

One consequence of the Peter--Weyl Theorem for compact groups (see, e.g., Theorem 1.12 in \cite{Knapp2001}) is that every unitary representation of compact groups is decomposable into irreducible representations in the following manner. Let $G$ be a compact group, let $d$ be a positive integer, and let $\rho:G\rightarrow \U{d}$ be a unitary representation of $G$. There exists a choice of unitary matrix $U\in\U{d}$ along with positive integers $m_1,\dots,m_p$ and~$n_1,\dots,n_p$ satisfying $d=m_1n_1+\cdots+m_pn_p$ such that the representation decomposes as 
\begin{equation}
 U\rho(g)U^* = \bigl(\I_{m_1}\otimes \rho_1(g)\bigr)\oplus\cdots\oplus\bigl(\I_{m_p}\otimes \rho_p(g)\bigr)
\end{equation}
for some inequivalent irreducible representations $\rho_1:G\rightarrow\U{n_1},\dots,\rho_p:G\rightarrow\U{n_p}$.  The \emph{commutant algebra} of this representation---the set of all matrices in $\M_d$ that commute with $\rho(g)$ for every group element $g\in G$---is a subalgebra of~$\M_d$ and may be expressed as
\begin{equation}\label{eq:commutantalgebra}
 \op{comm}(\{\rho(g)\,:\, g\in G\})  = U^*\bigl((\M_{m_1}\otimes\I_{n_1})\oplus \cdots\oplus (\M_{m_p}\otimes\I_{n_p})\bigr)U,
\end{equation}
where $U\in\U{d}$ is the same unitary matrix as above. (See, e.g., Section 1.7 of \cite{Sagan2001}.) Let $\Phi:\M_d\rightarrow\M_d$ be the corresponding twirling channel for this representation, which is defined as
\begin{equation}
 \Phi(X) = \int_{G}\mathrm{d}\mu(g)\, \rho(g)X\rho(g)^*
\end{equation}
for every $X\in\M_d$, where $\mu$ is the Haar measure on $G$. As the twirling channel must be the projection of~$\M_d$ onto the commutant algebra of the representation given in~\eqref{eq:commutantalgebra}, the channel $\Phi$ may be alternatively expressed as
\begin{equation}
 \Phi(X) = U^*\Bigl(\bigl((\I_{\M_{m_1}}\otimes\Omega_{n_1})\oplus\cdots\oplus(\I_{\M_{m_p}}\otimes\Omega_{n_p})\bigr)(UXU^*)\Bigr)U
\end{equation}
for every $X\in\M_d$, where $U\in\U{d}$ is the same unitary matrix as chosen above. 

\subsection{Trace vectors and von Neumann algebras}
 Let $d$ be a positive integer and let $\mathcal{A}\subset\M_d$ be a subset. We shall say that a vector $v\in\complex^d$ is a \emph{trace vector} for $\mathcal{A}$ if it holds that 
 \begin{equation}
  \langle v,Av\rangle = \Tr(A)
 \end{equation}
 for every $A\in\mathcal{A}$. Properties of trace vectors were first investigated in \cite{Pereira2003} (although the term `trace vector' was defined slightly differently there). If $\I_d\in\mathcal{A}$ and $v\in\complex^d$ is a trace vector for $\mathcal{A}$, it must be the case that
 \begin{equation}
  \lVert v\rVert = \sqrt{\ip{v}{v}} = \sqrt{\ip{v}{\I_d v}}= \sqrt{\Tr(\I_d)} = \sqrt{d}.
 \end{equation}
 A subalgebra $\mathcal{A}\subset\M_d$ is a \emph{von Neumann algebra} if it is both unital and self-adjoint (that is, if it satisfies both $\I_d\in\mathcal{A}$ and $\mathcal{A}^*=\mathcal{A}$). A useful characterization of von Neumann algebras in $\M_d$ was provided in \cite{Pereira2003}, which we summarize below.

\begin{theorem}[Pereira]\label{thm:rajesh}
  Let $d$ be a positive integer and let $\mathcal{A}\subset\M_d$ be a von Neumann algebra. The following statements are equivalent.
  \begin{enumerate}

 \item There exists a vector $v\in\complex^d$ that is a trace vector for $\mathcal{A}$.
 \item There exists an orthogonal basis of~$\complex^d$ such that each element of the basis is a trace vector for $\mathcal{A}$. 
   \item There exists a positive integer $p$, positive integers integers $m_1,\dots,m_p$ and~$n_1,\dots,n_p$ that satisfy $m_k\geq n_k$ for each $k\in\{1,\dots,p\}$, and a unitary matrix $U\in\U{d}$ such that
   \[
    U\mathcal{A} U^* = (\I_{m_1}\otimes\M_{n_1})\oplus\cdots\oplus(\I_{m_p}\otimes\M_{n_p}).
   \]
  \end{enumerate}

\end{theorem}

\section{New results}

\subsection{Condition for a channel to be mixed unitary in terms of trace vectors}

Making use of the notion of trace vectors, we may restate the necessary and sufficient conditions for a channel to be mixed unitary from Theorem \ref{thm:mixedunitarycomplement} in terms of the image of a complementary channel restricted to traceless matrices having an orthonormal basis of trace vectors. 

\begin{theorem}\label{thm:newcharacterization}
 Let~$n$ be a positive integer and let $\Phi:\M_n\rightarrow\M_n$ be a channel having Choi rank $r$. For each integer $N\geq r$, the following statements are equivalent.
 \begin{enumerate}
  \item The channel $\Phi$ is mixed unitary with mixed-unitary rank at most $N$.
  \item For some choice of channel $\Psi:\M_n\rightarrow\M_N$ that is complementary to $\Phi$, there exists an orthogonal basis of trace vectors for the set $\{\Psi(X) :\Tr(X)=0\}\subset\M_N$. 
 \end{enumerate}

\end{theorem}

\begin{proof}
  First suppose that (2) holds. Let $\Psi:\M_n\rightarrow\M_N$ be a channel complementary to $\Phi$ and let $\{v_1,\dots,v_N\}\subset\complex^N$ be an orthogonal basis of trace vectors for the set $\{\Psi(X) :\Tr(X)=0\}$. Define a unitary matrix $U\in\U{N}$ as 
 \begin{equation}
   U = \sum_{k=1}^N\frac{1}{\norm{v_k}}e_kv_k^*
 \end{equation}
and define a linear mapping $\Xi:\M_n\rightarrow\M_N$ as $\Xi(X) = U\Psi(X)U^*$ for every $X\in\M_n$. This mapping is also a channel that is complementary to $\Phi$. Moreover, for each traceless matrix $X\in\M_n$, the diagonal entries of~$\Xi(X)$ are equal to
\begin{equation}
 \ip{e_k}{\Xi(X)e_k} = \frac{1}{\norm{v_k}^2}\ip{v_k}{\Psi(X)v_k} = \frac{1}{\norm{v_k}^2}\Tr(\Psi(X)) = \frac{1}{\norm{v_k}^2}\Tr(X) = 0
\end{equation}
for each index $k\in\{1,\dots,n\}$, where we make use of the fact that $\Psi$ is trace preserving and that $v_k$ is a trace vector for $\{\Psi(X) :\Tr(X)=0\}$. Therefore $\Xi(X)$ has vanishing diagonal for each traceless matrix $X\in\M_n$. That $\Phi$ is mixed unitary with mixed-unitary rank at most $N$ now follows from Theorem \ref{thm:mixedunitarycomplement}. On the other hand, if statement (1) holds, then, by Theorem \ref{thm:mixedunitarycomplement}, there exists a choice of channel $\Psi:\M_n\rightarrow\M_N$ that is complementary to $\Phi$ such that $\Phi(X)$ has vanishing diagonal for every traceless matrix $X\in\M_n$. For this choice of complementary channel, it is evident that the vectors $e_1,\dots,e_N$ form an orthogonal basis of trace vectors in $\complex^N$ for the set $\{\Psi(X) :\Tr(X)=0\}$, as it holds that 
\begin{equation}
 \ip{e_k}{\Psi(X)e_k} = 0 = \Tr(X) = \Tr(\Psi(X))
\end{equation}
for each $k\in\{1,\dots,N\}$ and every traceless matrix $X\in\M_n$, where in the final equality we make use of the fact that $\Psi$ is trace preserving.
\end{proof}

If the image of the complementary channel $\Psi:\M_n\rightarrow\M_N$ in the proof of Theorem \ref{thm:newcharacterization} can be shown to be embedded inside of a von Neumann algebra, then the characterization of trace vectors for von Neumann algebras in Theorem \ref{thm:rajesh} can be used to bound the mixed-unitary rank of the channel, as the following corollary indicates.
\begin{corollary}\label{cor:imPsiinA}
Let~$n$ be a positive integer, let $\Phi:\M_n\rightarrow\M_n$ be a channel, and let $\Psi:\M_n\rightarrow\M_N$ be a channel complementary to $\Phi$ for some positive integer $N$. If it holds that $\op{im}(\Psi)\subset\mathcal{A}$ for some von Neumann algebra $\mathcal{A}\subset\M_N$ having the form 
\begin{equation}\label{eq:A in direct sum of IM}
 \mathcal{A} = (\I_{m_1}\otimes\M_{n_1})\oplus\cdots\oplus(\I_{m_p}\otimes\M_{n_p})
\end{equation}
for some positive integers $m_1,\dots,m_p$ and $n_1,\dots,n_p$ satisfying $m_k\geq n_k$ for each $k\in\{1,\dots,p\}$, then the channel $\Phi$ is mixed unitary with mixed-unitary rank at most $N$. 
\end{corollary}

\subsection{Direct sums of completely depolarizing channels have minimal mixed-unitary rank}

 Let $p$ be a positive integer, let $m_1,\dots,m_p$ and~$n_1,\dots,n_p$ be positive integers, define the number $d=m_1n_1+\cdots+m_pn_p$, and let $\Phi:\M_d\rightarrow\M_d$ be the channel defined as 
 \begin{equation}\label{eq:Phimn}
  \Phi(X) = \bigl((\I_{\M_{m_1}}\otimes\Omega_{n_1})\oplus\cdots\oplus(\I_{\M_{m_p}}\otimes\Omega_{n_p})\bigr)(X)
 \end{equation}
 for every $X\in\M_d$. It is evident that the Choi rank of the channel in~\eqref{eq:Phimn} is equal to $N$, where one defines the integer $N=n_1^2+\cdots+n_p^2$. Moreover, this channel is itself a mixed-unitary channel as it can be expressed as a direct sum of mixed-unitary channels. Using the results of the previous section, we now show that $\Phi$ has mixed-unitary rank also equal to $N$ (and thus it is said to have \emph{minimal} mixed-unitary rank). To see this, recall from the observations in Section \ref{sec:prelim} that one channel that is complementary to $\Phi$ is the channel $\Psi:\M_d\rightarrow\M_N$ defined as  
 \begin{equation}
 \Psi\bigl((X_1\otimes Y_1)\oplus\cdots\oplus(X_p\otimes Y_p)\bigr) = \left(\frac{\Tr(X_1)}{n_1} \I_{n_1}\otimes Y_1\right)\oplus\cdots\oplus\left(\frac{\Tr(X_p)}{n_p} \I_{n_p}\otimes Y_p\right)
\end{equation}
for every $X_1\in\M_{m_1},\dots,X_p\in\M_{m_p}$ and $Y_1\in\M_{n_1},\dots,Y_p\in\M_{n_p}$. For this channel, it is evident that 
\begin{equation}
 \op{im}(\Psi) = (\I_{n_1}\otimes\M_{n_1})\oplus\cdots \oplus(\I_{n_p}\otimes\M_{n_p}). 
\end{equation}
As $\op{im}(\Psi)$ is a von Neumann algebra of the form in~\eqref{eq:A in direct sum of IM}, by Corollary \ref{cor:imPsiinA} we see that $\Phi$ also has mixed-unitary rank equal to $N$.

\subsection{Twirling channels have minimal mixed-unitary rank}\label{sec:twirlingrank}

It is now straightforward to verify the central claim of this work, which states that every twirling channel has minimal mixed-unitary rank. Recall from Section \ref{sec:representation} that every twirling channel is of the following form.
Let $d$ be a positive integer, let $G$ be a compact group, and let $\rho:G\rightarrow\U{d}$ be a unitary representation of $G$. There exists a positive integer $p$, positive integers $m_1,\dots,m_p$ and~$n_1,\dots,n_p$, and a choice of unitary matrix $U\in\U{d}$ (where one defines $d=m_1n_1+\cdots+m_pn_p$) such that 
\begin{equation}
 U\Phi_\rho(U^*XU)U^* = \bigl((\I_{\M_{m_1}}\oplus\Omega_{n_1})\oplus\cdots\oplus (\I_{\M_{m_p}}\oplus\Omega_{n_p})\bigr)(X)
\end{equation}
for every $X\in\M_d$. (That is, $\Phi_\rho$ is the projection onto the commutant algebra of $\rho$.) From this observation, along with the observations in the previous section, it is now evident that $\Phi_\rho$ is a mixed-unitary channel with minimal mixed-unitary rank. In particular, this twirling channel has Choi rank and mixed-unitary rank both equal to~$n_1^2+\cdots+n_p^2$. This number is also equal to the dimension of the von Neumann algebra generated by the representation $\rho$, which is the set $\mathcal{A}_\rho\subset\M_d$ defined as
\begin{equation}
 \mathcal{A}_\rho=\op{comm}(\op{comm}(\{\rho(g):g\in G\})).
\end{equation}
Note that the \emph{commutant algebra} of the representation $\rho$ is the von Neumann algebra of~$\M_d$ that contains all elements of $\M_d$ that commute with $\rho(g)$ for every $g\in G$. The commutant algebra of $\rho$ can be decomposed as
\begin{equation}
 \op{comm}(\{\rho(g):g\in G\}) = U^*\Bigl((\M_{m_1}\otimes\I_{n_1})\oplus\cdots\oplus(\M_{m_p}\otimes\I_{n_p})\Bigr)U,
\end{equation}
where $U\in\U{d}$ is the same unitary matrix as chosen above. The von Neumann algebra generated by $\rho$ is the double commutant of the representation and can therefore be expressed as
\begin{equation}
 U\mathcal{A}_\rho U^* = (\I_{m_1}\otimes\M_{n_1})\oplus\cdots\oplus(\I_{m_p}\otimes\M_{n_p}).
\end{equation}
It evident that the dimension of the von Neumann algebra is equal to $\dim(\mathcal{A}_\rho) = N$, where one defines integer $N = n_1^2+\cdots+n_p^2$. These observations imply the existence of a probability distribution $(p_1,\dots,p_N)$ and unitary matrices $U_1,\dots,U_N\in\U{d}$ such that $\Phi_\rho$ can be expressed as 
\begin{equation} \label{eq:Phirho}
 \Phi_\rho(X) = \sum_{k=1}^N p_k U_kXU_k
\end{equation}
for every $X\in\M_d$. We collect these ideas in the following theorem.

\begin{theorem}\label{thm:maintheorem}
 Every finite-dimensional twirling channel has minimal mixed-unitary rank, having mixed-unitary rank equal to the dimension of the von Neumann algebra of the underlying representation.
\end{theorem}

The proof of Theorem \ref{thm:maintheorem} essentially makes use of Theorem \ref{thm:rajesh}. The proof of Theorem \ref{thm:rajesh} in \cite{Pereira2003} is constructive, which allows us to \emph{explicitly} construct minimal mixed-unitary decompositions for twirling channels. The remainder of this paper is dedicated to these explicit constructions and examples of twirling channels. 

\section{Explicit construction of minimal mixed-unitary decompositions}
\label{sec:explicit}
In this section we construct a minimal mixed-unitary decompositions for channels of the form in \eqref{eq:Phimn}. The following construction is based on the proof of Theorem \ref{thm:rajesh} (see Theorem 3.2.4 in in \cite{Pereira2003}). Further details and a proof of this construction are provided in Appendix \ref{app:construction}. This section concludes with an explicit example in the simplest non-trivial case.

Let $p$ be a positive integer, let $m_1,\dots,m_p$ and $n_1,\dots,n_p$ be positive integers, define the integer $d= m_1n_1 + \cdots + m_pn_p$, and consider the channel $\Phi:\M_d\rightarrow\M_d$ defined as
\begin{equation}\label{eq:Phimn2}
 \Phi = (\I_{\M_{m_1}}\otimes\Omega_{n_1})\oplus \cdots \oplus(\I_{\M_{m_p}}\otimes\Omega_{n_p}).
\end{equation}
This channel has both Choi rank and mixed-unitary rank equal to $N=n_1^2+\cdots+n_p^2$. For each $k\in\{1,\dots,N\}$ and $\ell\in\{1,\dots,p\}$, define a matrix $U_{k,\ell}\in\M_{n_\ell}$ as
\begin{equation}
 U_{k,\ell} = \frac{1}{n_\ell}\sum_{a,b=1}^{n_\ell}\left(\sum_{c=1}^{n_\ell}\exp\left(2\pi i\frac{c(b-a)}{n_\ell}\right)\exp\left(2\pi i\cdot k\frac{N_\ell+(a-1)n_\ell+c}{N}\right)\right)E_{a,b},
\end{equation}
where one defines the numbers $N_1=0$ and $
 N_\ell = n_1^2+\cdots + n_{\ell-1}^2$ for each $\ell\in\{2,\dots,p\}$. Now define matrices $U_1,\dots,U_N\in\M_d$  as
\begin{equation}\label{eq:UK1toUKp}
 U_k = (\I_{m_1}\otimes U_{k,1})\oplus\cdots\oplus(\I_{m_p}\otimes U_{k,p})
\end{equation}
for each $k\in\{1,\dots,N\}$. It may be verified (see Appendix \ref{app:construction}) that each of the matrices $U_1,\dots,U_N$ is unitary and that
\begin{equation}\label{eq:PhiNMU}
 \Phi(X) = \frac{1}{N}\sum_{k=1}^N U_kXU_k^*
\end{equation}
holds for every $X\in\M_{d}$. As the channel $\Phi$ is known to have mixed-unitary rank equal to $N$, the expression in \eqref{eq:PhiNMU} is therefore a minimal mixed-unitary decomposition of the channel defined in~\eqref{eq:Phimn2}. We also remark that the expression in \eqref{eq:PhiNMU} means that $\Phi$ can be expressed as the \emph{average} of the $N$ unitary channels defined by the unitary matrices $U_1,\dots,U_N\in\U{d}$.

\begin{example}
This example uses the construction outlined above to construct a minimal mixed-unitary decomposition of the smallest non-trivial example of a channel of the form in \eqref{eq:Phimn2}.  This is the channel $\Phi:\M_3\rightarrow\M_3$ defined as
\begin{equation}
 \Phi = \Omega_{2}\oplus\Omega_1
\end{equation}
 (in which case one has $p=2$, $n_1=2$, and $m_1=m_2=n_2=1$). This channel has both Choi rank and mixed-unitary rank equal to 5. For this example, the resulting matrices from the expression in \eqref{eq:UK1toUKp} are the matrices $U_1,U_2,U_3,U_4,U_5\in\M_3$ defined as
 \begin{align*}
 U_1 &= \left(
\begin{array}{ccc}
 -\frac{1}{4} +i\frac{\sqrt{5+2\sqrt{5}}}{4}  & -\frac{\sqrt{5}}{4}-i\frac{\sqrt{5-2 \sqrt{5}}}{4} & 0 \\
 \frac{\sqrt{5}}{4}-i\frac{\sqrt{5-2 \sqrt{5}}}{4}  & -\frac{1}{4}-i\frac{\sqrt{5+2\sqrt{5}}}{4}  & 0 \\
 0 & 0 & 1 \\
\end{array}
\right)\\
U_2&=\left(
\begin{array}{ccc}
 -\frac{1}{4}-i\frac{\sqrt{5-2 \sqrt{5}}}{4}   & \frac{\sqrt{5}}{4}-i\frac{\sqrt{5+2\sqrt{5}}}{4}  & 0 \\
 -\frac{\sqrt{5}}{4}-i\frac{\sqrt{5+2\sqrt{5}}}{4}   & -\frac{1}{4}+i\frac{\sqrt{5-2 \sqrt{5}}}{4}  & 0 \\
 0 & 0 & 1 \\
\end{array}
\right)\\
U_3&=\left(
\begin{array}{ccc}
 -\frac{1}{4}+i\frac{\sqrt{5-2 \sqrt{5}}}{4}   & \frac{\sqrt{5}}{4}+i\frac{\sqrt{5+2\sqrt{5}}}{4}   & 0 \\
 -\frac{\sqrt{5}}{4}+i\frac{\sqrt{5+2\sqrt{5}}}{4}  & -\frac{1}{4}-i\frac{\sqrt{5-2 \sqrt{5}}}{4}  & 0 \\
 0 & 0 & 1 \\
\end{array}
\right)\\
U_4&=\left(
\begin{array}{ccc}
 -\frac{1}{4}-i\frac{\sqrt{5+2\sqrt{5}}}{4}  & -\frac{\sqrt{5}}{4}+i\frac{\sqrt{5-2 \sqrt{5}}}{4}  & 0 \\
 \frac{\sqrt{5}}{4}+i\frac{\sqrt{5-2 \sqrt{5}}}{4}  & -\frac{1}{4}+i\frac{\sqrt{5+2\sqrt{5}}}{4}  & 0 \\
 0 & 0 & 1 \\
\end{array}
\right)\\
U_5&=\left(
\begin{array}{ccc}
 1 & 0 & 0 \\
 0 & 1 & 0 \\
 0 & 0 & 1 \\
\end{array}
\right).
\end{align*}
These matrices are unitary and satisfy
\begin{equation}
 (\Omega_2\oplus\Omega_1)(X) = \frac{1}{5}\sum_{k=1}^5 U_kXU_k^*
\end{equation}
for every $X\in\M_3$. Interestingly, these matrices also satisfy 
\begin{equation}
 \ip{U_k}{U_{k'}} = \left\{\begin{array}{ll}3 & \text{if }k=k'\\ \frac{1}{2} & \text{if }k\neq k'\end{array}\right.
\end{equation}
for each pair of indices $k,k'\in\{1,\dots,5\}$. (That is, these matrices are \emph{equiangular} in $\M_3$.)
\end{example}

\section{Examples of twirling channels}

In this section we explore some examples of twirling channels and their minimal mixed-unitary decompositions.

\subsection{Permutation-twirling channel}

Let $n$ be a positive integer denote the symmetric group of order~$n$ by $\op{Sym}(n)$. For each permutation $\pi\in\op{Sym}(n)$, one defines the permutation matrix $V_\pi\in\M_d$ as
\begin{equation}
 V_\pi = \sum_{k=1}^n E_{\pi(k),k}
\end{equation}
This is naturally a unitary representation of $\op{Sym}(n)$, but it is not irreducible. This representation decomposes as a direct sum of the 1-dimensional \emph{trivial representation} and the $(n-1)$-dimensional \emph{natural representation} of $\op{Sym(n)}$, each of which are irreducible. (See, e.g., Section 4.4.3 of \cite{Goodman2009}.) One may define the \emph{permutation-twirling channel} $\Phi:\M_n\rightarrow\M_n$ as
\begin{equation}\label{eq:permtwirl}
\Phi_{\op{Sym}(n)}(X) = \frac{1}{n!}\sum_{\pi\in\op{Sym}(n)} V_\pi XV_\pi^*
\end{equation}
for every $X\in\M_n$. From the observations above regarding the decomposition of this representation (and from the results in Section \ref{sec:twirlingrank}), we see that the permutation-twirling channel of dimension $n$, as defined in \eqref{eq:permtwirl}, has mixed-unitary rank equal to its Choi rank, which is equal to 
\begin{equation}
 \op{rank}(J(\Phi_{\op{Sym}(n)})) = (n-1)^2 + 1.
\end{equation}
In fact, the commutator of the set of $n\times n$ permutation matrices can be expressed as
\begin{equation}
 \op{comm}(\{V_\pi:\pi\in\op{Sym}(n)\}) = \op{span}\Bigl\{\frac{1}{n}J_n,\, \I_n-\frac{1}{n}J_n\Bigr\},
\end{equation}
where $J_n$ denotes the $n\times n$ matrix whose entries are all equal to 1. The matrices $J_n/n$ and $\I_n-J_n/n$ are projection matrices having ranks $1$ and $n-1$ respectively. The permutation-twirling channel may therefore be alternatively expressed as
\begin{equation}
 \Phi_{\op{Sym}(n)}(X) = \Bigip{\frac{1}{n}J_n}{X} \frac{1}{n}J_n + \Bigip{\I_n - \frac{1}{n} J_n}{X} \frac{1}{n-1}\left(\I_n - \frac{1}{n} J_n\right).
\end{equation}
The channel $\Phi_{\op{Sym}(n)}$ has both Choi rank and mixed-unitary rank equal to $(n-1)^2+1$ by Theorem~\ref{thm:maintheorem}. Moreover, the construction in Section \ref{sec:explicit} allows us to explicitly construct a minimal mixed-unitary decomposition for $\Phi_{\op{Sym}(n)}$. To do so, first define a unitary matrix $U\in\U{n}$ as
\begin{equation}
 U= \frac{1}{\sqrt{n}}\sum_{a,b=1}^n \exp\left(2\pi i\frac{ab}{n}\right)E_{a,b}
\end{equation}
which satisfies
\begin{equation}
 \frac{1}{n}U J_n U^* = \begin{pmatrix} 0 & \cdots & 0 & 0\\ \vdots & \ddots & \vdots & \vdots \\
 0 & \cdots & 0 & 0\\
 0 & \cdots & 0 & 1\end{pmatrix}
 \qquad\text{and}\qquad
  U \left(\I_n-\frac{1}{n}J_n\right) U^* = \begin{pmatrix} 1 & \cdots & 0 & 0\\ \vdots & \ddots & \vdots & \vdots \\
 0 & \cdots & 1 & 0\\
 0 & \cdots & 0 & 0\end{pmatrix}
\end{equation}
and thus the permutation twirling channel may be expressed as
\begin{equation}
 \Phi_{\op{Sym}(n)}(X) = U^*\bigl((\Omega_{n-1}\oplus\Omega_1)(UXU^*)\bigr)U.
\end{equation}
One may therefore construct a minimal mixed-unitary decomposition for $\Phi_{\op{Sym}(n)}$ as follows. Following the construction in Section \ref{sec:explicit}, one may define the matrices $U_1,\dots,U_{(n-1)^2+1}\in\M_n$ as
\begin{equation}
 U_k = \left(\frac{1}{n-1}\sum_{a,b,c=1}^{n-1}\exp\left(2\pi i \frac{c(b-a)}{n-1}\right)\exp\left(2\pi i\cdot k \frac{(a-1)(n-1)+c}{(n-1)^2+1}\right)E_{a,b}\right) + E_{n,n}
\end{equation}
for each $k\in\{1,\dots,(n-1)^2+1\}$ such that
\begin{equation}
 \Phi_{\op{Sym}(n)}(X) = \frac{1}{(n-1)^2+1}\sum_{k=1}^{(n-1)^2+1}(U^*U_kU)X(U^*U_kU)^*
\end{equation}
holds for every $X\in\M_n$.

\subsection{Werner twirling channel}

Let $n$ be a positive integer. The \emph{Werner twirling channel} of dimension $n$ is the linear mapping $\Xi:\M_{n^2}\rightarrow\M_{n^2}$ defined as 
\begin{equation}
\Xi(X) = \int_{\U{n}} \d\mu(U)\,(U\otimes U)X(U\otimes U)^* 
\end{equation}
for every $X\in\M_{n^2}$. The representation $U\mapsto U\otimes U$ decomposes into the direct sum of irreducible representations acting on the symmetric and anti-symmetric subspaces of $\complex^n\otimes\complex^n$. Define now a pair of projection matrices $\Pi_0,\Pi_1\in\M_{n^2}$ as
\begin{equation}
 \Pi_0 = \frac{1}{2}\I_n\otimes\I_n + \frac{1}{2}\sum_{j,k=1}^n E_{j,k}\otimes E_{k,j}
 \qquad\text{and}\qquad
  \Pi_0 = \frac{1}{2}\I_n\otimes\I_n - \frac{1}{2}\sum_{j,k=1}^n E_{j,k}\otimes E_{k,j},
\end{equation}
which are the projections onto the symmetric and anti-symmetric subspaces respectively and have ranks given by
\[
 \op{rank}(\Pi_0) = \binom{n+1}{2}\qquad\text{and}\qquad \op{rank}(\Pi_1) = \binom{n}{2}.
\]
The Werner twirling channel may alternatively be expressed as 
\begin{equation}
 \Xi(X) = \frac{1}{\binom{n+1}{2}}\ip{\Pi_0}{X}\Pi_0 + \frac{1}{\binom{n}{2}}\ip{\Pi_1}{X}\Pi_1.
\end{equation}
It follows from Theorem \ref{thm:maintheorem} that $\Xi$ has both Choi rank and mixed-unitary rank equal to 
\begin{equation}
\op{rank}(J(\Xi)) = \op{rank}(\Pi_0)^2 + \op{rank}(\Pi_1)^2  =\binom{n+1}{2}^2 + \binom{n}{2}^2= \frac{n^4+n^2}{2}.
\end{equation}
Moreover, for any choice of unitary matrix $U\in\U{n^2}$ such that 
\begin{equation}
 U\Pi_0U^* = \I_{\binom{n+1}{2}} \oplus 0 \qquad\text{and}\qquad U\Pi_1U^* = 0\oplus\I_{\binom{n}{2}} ,
\end{equation}
the Werner twirling channel may be expressed as
\begin{equation}
 \Xi(X) = U\bigl(\bigl(\Omega_{\binom{n+1}{2}}\oplus\Omega_{\binom{n}{2}}\bigr)(UXU^*)\bigr)U^*
\end{equation}
for every $X\in\M_{n^2}$. Following the construction in Section \ref{sec:explicit}, there is a choice of unitary matrices $U_1,\dots,U_{(n^4+n^2)/2}\in\U{n^2}$ such that the Werner twirling channel may be expressed as
\begin{equation}
 \Xi(X) = \frac{2}{n^4+n^2} \sum_{k=1}^{\frac{n^4+n^2}{2}} U_kXU_k^*.
\end{equation}
for every $X\in\M_{n^2}$.

\section{Acknowledgements}

MG is supported by the Natural Sciences and Engineering Research Council (NSERC) of Canada, by the Canadian Institute for Advanced Research (CIFAR), and through funding provided to IQC by the Government of Canada.
JL is supported by funding provided by IQC and the University of Guelph.

\appendix

\section{Explicit minimal mixed-unitary decompositions}
\label{app:construction}

This appendix provides the details of and the proof for the explicit construction from Section \ref{sec:explicit}. The following construction is based on the work of Pereira (see the proof of Theorem 3.2.4 in \cite{Pereira2003}).

Let $p$ be a positive integer, let $m_1,\dots,m_p$ and $n_1,\dots,n_p$ be positive integers, define the integer $d= m_1n_1 + \cdots + m_pn_p$, and consider the channel $\Phi:\M_d\rightarrow\M_d$ defined as
\begin{equation}
 \Phi = (\I_{\M_{m_1}}\otimes\Omega_{n_1})\oplus \cdots \oplus(\I_{\M_{m_p}}\otimes\Omega_{n_p}).
\end{equation}
This channel has both Choi rank and mixed-unitary rank equal to $N=n_1^2+\cdots+n_p^2$. We now proceed with an explicit construction of a minimal mixed-unitary decomposition. First note that one may partition the set of integers from $1$ to $N$ as
\begin{align}\label{eq:Npartition}
 \{1,\dots,N\} 
 &=\bigcup_{\ell=1}^p \bigl\{N_{\ell}+ j \, \big|\, 1\leq j\leq n_\ell^2\bigr\}\nonumber\\
 &= \bigcup_{\ell=1}^p\bigl\{N_{\ell} + (a-1)n_\ell+b\, \big|\, 1\leq a,b\leq n_\ell\bigr\},
\end{align}
where one defines the integers $N_1,\dots,N_{p}$ as $N_1=0$ and 
\begin{equation}
  N_{\ell}=n_1^2 + \cdots + n_{\ell-1}^2
\end{equation}
for each $\ell\in\{2,\dots,p\}$. For each $k\in\{1,\dots,N\}$ and $\ell\in\{1,\dots,p\}$, define a matrix $U_{k,\ell}\in\M_{n_\ell}$ as
\begin{equation}
 U_{k,\ell} = \frac{1}{n_\ell}\sum_{a,b=1}^{n_\ell}\left(\sum_{c=1}^{n_\ell}\exp\left(2\pi i\frac{c(b-a)}{n_\ell}\right)\exp\left(2\pi i\cdot k\frac{N_\ell+(a-1)n_\ell+c}{N}\right)\right)E_{a,b}.
\end{equation}
It may be verified that each of these matrices is unitary, as
\begin{align*}
 U_{k,\ell}U_{k,\ell}^* 
    & =\frac{1}{n_\ell^2} \sum_{a,b,c=1}^{n_\ell}\sum_{a',b',c'=1}^{n_\ell}\exp\left(2\pi i \frac{cb-ca-c'b'+c'a'}{n_\ell}\right)\exp\left(2\pi i \cdot k\frac{(a-a')n_\ell+c-c'}{N}\right) E_{a,b}E_{b',a'}\\
    & = \frac{1}{n_\ell^2} \sum_{a,a',c,c'=1}^{n_\ell}\sum_{b=1}^{n_\ell}\exp\left(2\pi i \frac{b(c-c')}{n_\ell}\right)\exp\left(2\pi i \frac{c'a'-ca}{n_\ell}\right)\exp\left(2\pi i\cdot k \frac{(a-a')n_\ell+c-c'}{N}\right) E_{a,a'}\\
    & = \frac{1}{n_\ell} \sum_{a,a'=1}^{n_\ell}\sum_{c=1}^{n_\ell}\exp\left(2\pi i \frac{c(a'-a)}{n_\ell}\right)\exp\left(2\pi i \frac{k\bigl((a-a')n_\ell\bigr)}{N}\right) E_{a,a'}\\
    & = \sum_{a=1}^{n_\ell} E_{a,a}\\
    & = \I_{n_\ell},
\end{align*}
where, for each $\ell\in\{1,\dots,p\}$, we use the fact that
\begin{equation}
 \sum_{c=1}^{n_\ell} \exp\left(2\pi i \frac{c(a-b)}{n_\ell}\right)= \left\{\begin{array}{ll}  n_\ell & \text{if }a=b\\ 0 & \text{otherwise} 
\end{array}
\right.
\end{equation}
holds for all pairs of numbers $a,b\in\{1,\dots,n_\ell\}$. For each $\ell\in\{1,\dots,p\}$, the collection of matrices $U_{1,\ell},\dots,U_{N,\ell}\in\U{n_\ell}$ comprises a mixed-unitary decomposition for the completely depolarizing channel $\Omega_{n_\ell}$. Indeed, for every $X\in\M_{n_\ell}$, one has that
\begin{align*}
 \frac{1}{N}\sum_{k=1}^N U_{k,\ell}XU_{k,\ell}^* 
   &= \frac{1}{Nn_\ell^2}\sum_{j=1}^N\sum_{a,b,c=1}^{n_\ell}\sum_{a',b',c'=1}^{n_\ell}\exp\left(2\pi i \frac{cb-ca-c'b'+c'a'}{n_\ell}\right)\\
   &\hspace{2.5in}\times\exp\left(2\pi i\cdot k \frac{(a-a')n_\ell+c-c'}{N}\right) E_{a,b}XE_{b',a'}\\
   &=\frac{1}{Nn_\ell^2}\sum_{a,b,c=1}^{n_\ell}\sum_{a',b',c'=1}^{n_\ell}\langle E_{b,b'},X\rangle\exp\left(2\pi i \frac{cb-ca-c'b'+c'a'}{n_\ell}\right)\\
   &\hspace{2.5in}\times\sum_{k=1}^N\exp\left(2\pi i\cdot k\frac{(a-a')n_\ell+c-c'}{N}\right) E_{a,a'}\\
   &=\frac{1}{n_\ell^2}\sum_{a,b,b'=1}^{n_\ell}\sum_{c=1}^{n_\ell}\langle E_{b,b'},X\rangle\exp\left(2\pi i \frac{c(b-b')}{n_\ell}\right) E_{a,a}\\
   &=\frac{1}{n_\ell}\sum_{a,b=1}^{n_\ell}\langle E_{b,b},X\rangle E_{a,a}\\
   &=\frac{\Tr(X)}{n_\ell}\I_{n_\ell}\\
   & = \Omega_{n_\ell}(X),
\end{align*}
where we use the fact that
\begin{equation}
 \sum_{j=1}^N\exp\left(2\pi i \cdot k\frac{(a-a')n_\ell+c-c'}{N}\right) = \left\{\begin{array}{ll}  N & \text{if }a=a'\text{ and } c=c'\\ 0 & \text{otherwise} 
\end{array}
\right.
\end{equation}
holds for all choices of numbers $a,a',c,c'\in\{1,\dots,n_\ell\}$. Moreover, making use of the partition in~\eqref{eq:Npartition}, for every pair of indices $\ell,\ell'\in\{1,\dots,p\}$ one has that
\begin{multline}\label{eq:sumjtoNkkaacc}
 \sum_{k=1}^N\exp\left(2\pi i\cdot k \frac{(N_\ell +(a-1)n_\ell + c) - (N_{\ell'} +(a'-1)n_{\ell'} + c')}{N}\right) \\= \left\{\begin{array}{ll}  N & \text{if }\ell=\ell',\text{ }a=a',\text{ and } c=c'\\ 0 & \text{otherwise} 
\end{array}
\right.
\end{multline}
for all choices of numbers $a,c\in\{1,\dots,n_\ell\}$ and $a',c'\in\{1,\dots,n_{\ell'}\}$. Consider now a collection of linear mappings $\Phi_{\ell,\ell'}:\M_{n_\ell,n_{\ell'}}\rightarrow\M_{n_\ell,n_{\ell'}}$ defined for each pair of indices $\ell,\ell'\in\{1,\dots,p\}$ that map $n_\ell\times n_{\ell'}$ matrices to $n_\ell\times n_{\ell'}$ matrices as
\begin{equation}
 \Phi_{\ell,\ell'}(X) = \frac{1}{N}\sum_{k=1}^N U_{k,\ell}XU_{k,\ell'}^*
\end{equation}
for every $X\in\M_{n_\ell,n_{\ell'}}$. In the case when $\ell=\ell'$, this map is precisely $\Phi_{\ell,\ell}=\Omega_{n_\ell}$. However, in the case when $\ell\neq \ell'$, one may make use of the equality in \eqref{eq:sumjtoNkkaacc} to see that
\begin{align*}
 \Phi_{\ell,\ell'}(X) 
  & = \frac{1}{Nn_\ell n_{\ell'}}\sum_{a,b,c=1}^{n_\ell}\sum_{a',b',c'=1}^{n_{\ell'}}
  \exp\left(2\pi i \frac{c(b-a)}{n_\ell}\right)\exp\left(2\pi i \frac{c'(a'-b')}{n_{\ell'}}\right)\cdot\\
   & \hspace{1in}\sum_{k=1}^N\exp
  \left(2\pi i \cdot k\frac{(N_\ell + (a-1)n_\ell+c)-(N_{\ell'}+(a'-1)n_{\ell'}+c')}{N}\right) E_{a,b}XE_{b',a'}\\ & =0
\end{align*}
holds for every $X\in\M_{n_\ell,n_{\ell'}}$ and thus $\Phi_{\ell,\ell'}=0$. Finally, define matrices $U_1,\dots,U_N\in\U{d}$  as
\begin{equation}
 U_k = (\I_{m_1}\otimes U_{k,1})\oplus\cdots\oplus(\I_{m_p}\otimes U_{k,p})
\end{equation}
for every $k\in\{1,\dots,N\}$, where we recall that $d=m_1n_1+\cdots+m_pn_p$. For every $X\in\M_{d}$, there exist matrices $X_{\ell,\ell'}\in\M_{m_\ell n_\ell,m_{\ell'}n_{\ell'}}$ for each pair of indices $\ell,\ell'\in\{1,\dots,p\}$ such that $X$ may be expressed in block form as
\begin{equation}
 X = \begin{pmatrix} X_{1,1} & \cdots & X_{1,p}\\ \vdots & \ddots& \vdots \\ X_{p,1} & \cdots & X_{p,p}\end{pmatrix}.
\end{equation}
For all such matrices, one has that
\begin{align*}
 \frac{1}{N}\sum_{k=1}^{N}U_k XU_k^*
 & = \frac{1}{N}\sum_{k=1}^{N}U_k\begin{pmatrix} X_{1,1}& \cdots & X_{1,p}\\ \vdots & \ddots& \vdots \\ X_{p,1} & \cdots & X_{p,p}\end{pmatrix}U_j^*
  \\
  &= \frac{1}{N}\sum_{k=1}^{N}\begin{pmatrix} (\I_{m_1}\otimes U_{k,1})X_{1,1}(\I_{m_1}\otimes U_{k,1})^* & \cdots & (\I_{m_1}\otimes U_{k,1})X_{1,p}(\I_{m_p}\otimes U_{k,p})^*\\ \vdots & \ddots& \vdots \\ (\I_{m_p}\otimes U_{k,p})X_{p,1}(\I_{m_1}\otimes U_{k,1})^* & \cdots & (\I_{m_p}\otimes U_{k,p})X_{p,p}(\I_{m_p}\otimes U_{k,p})\end{pmatrix}\\
  & = \begin{pmatrix}
       (\I_{\M_{m_1,m_{1}}}\otimes \Phi_{1,1})(X_{1,1})
        & \cdots & (\I_{\M_{m_1,m_{p}}}\otimes \Phi_{1,p})(X_{1,p})\\
         \vdots & \ddots & \vdots\\
         (\I_{\M_{m_p,m_{1}}}\otimes \Phi_{p,1})(X_{p,1}) & \cdots & (\I_{\M_{m_p,m_{p}}}\otimes \Phi_{p,p})(X_{p,p})
      \end{pmatrix}\\
  & = \begin{pmatrix}
        (\I_{\M_{m_1}}\otimes \Omega_{n_1})(X_{1,1})
         & \cdots & 0\\
          \vdots & \ddots & \vdots\\
          0 & \cdots & (\I_{\M_{m_p}}\otimes \Omega_{n_p})(X_{p,p})
       \end{pmatrix}\\
 & = \bigl((\I_{\M_{m_1}}\otimes \Omega_{n_1})\oplus\cdots\oplus (\I_{\M_{m_p}}\otimes \Omega_{n_p})\bigr)(X),
\end{align*}
as desired.

\bibliographystyle{alpha}

\bibliography{twirling}

\end{document}